\documentclass[final,5p,times]{elsarticle}

\usepackage{graphicx}
\usepackage{amssymb}
\usepackage{amsthm}
\usepackage{amsmath}

\usepackage{lineno}

\usepackage[table,xcdraw]{xcolor}
\usepackage{todonotes}
\usepackage{soul}
\usepackage{color}
\usepackage{xfrac}

\usepackage{algorithm}
\usepackage{algorithmic}

\usepackage{hyperref}

{\bf}{}
{}{}
{\bf}{}
{\bf}{}
\newtheorem{prop}{Proposition}{\bf}{}
{\bf}{}
{\bf}{}
{\bf}{}
{\bf}{}
{\bf}{}
\newtheorem{form}{Formulation}

\begin{document}

\begin{frontmatter}


\title{A branch-\&-price approach to the unrooted maximum agreement forest problem}



\author[a]{Martin Frohn\corref{1}\fnref{fn1}}

\author[a]{Steven Kelk\fnref{fn2}}

\author[a]{Simona Vychytilova\fnref{fn3}}

\address[a]{Department of Advanced Computing Sciences, Maastricht University, Paul Henri Spaaklaan 1, 6229 EN Maastricht, The Netherlands}

\cortext[1]{Corresponding author}
\fntext[fn1]{Email: \href{mailto:martin.frohn@maastrichtuniversity.nl}{martin.frohn@maastrichtuniversity.nl} (Martin Frohn)}
\fntext[fn2]{Email: \href{mailto:steven.kelk@maastrichtuniversity.nl}{steven.kelk@maastrichtuniversity.nl} (Steven Kelk)}
\fntext[fn3]{Email: \href{mailto:simivych@gmail.com}{simivych@gmail.com}
(Simona Vychytilova)}

\begin{abstract}
We propose the first branch-\&-price algorithm for the maximum agreement forest problem on
unrooted binary trees: given two unrooted $X$-labelled binary trees
we seek to partition $X$ into a minimum number of blocks such
that the induced subtrees are disjoint and have the same topologies in
both trees.
We provide a dynamic programming algorithm
for the weighted maximum agreement subtree problem to solve the pricing problem.
When combined with rigorous polynomial-time pre-processing
our branch-\&-price algorithm exhibits (beyond) state-of-the-art performance.
\end{abstract}


\end{frontmatter}

\section{Introduction}
Given a set of $n$ species $X=\{x_1,\dots,x_n\}$, a $X$-tree is a tree
$T=(V,E)$ whose $n$ leaves are bijectively labelled by $X$. We call $T$ \emph{rooted} if there exists a vertex $\rho\in V$ with in-degree $0$ and out-degree $2$ such that there exists a (unique) directed path from $\rho$ to every leaf of $T$. Otherwise, if $T$ is undirected we call it \emph{unrooted}. In this article we are mainly concerned with unrooted binary $X$-trees: each internal vertex has degree 3.

Such trees are a fundamental object of study in the field of phylogenetics, which seeks to reconstruct the evolutionary history of a set of contemporary species $X$ in the absence of information about extinct ancestors \cite{SempleSteel2003}. The literature on constructing these phylogenetic $X$-trees is vast. For a variety of methodological and biological reasons different inference methods can produce trees with different topologies, and then it is natural to wish to quantify in a formal sense how topologically different two trees are \cite{som2015causes}. Many such pairwise dissimilarity measures have been proposed. Here we focus on the well-studied \emph{maximum agreement forest} (MAF) model. Given two phylogenetic $X$-trees $T_1, T_2$ this model asks for $X$ to be partitioned into a minimum number of blocks such that in both $T_1$ and $T_2$ the subtrees induced by the blocks are disjoint, and that for each block the subtree induced in $T_1$ is isomorphic to that induced in $T_2$ (up to suppression of degree 2 nodes, and taking leaf-labels into account). Such a minimum-size partition is called a maximum agreement forest.
Notably, the MAF model is not just a dissimilarity measure: it turns out to fulfill an important role in the  detection of reticulate evolutionary phenomena, such as horizontal gene transfer, and the construction of phylogenetic networks (as opposed to trees) \cite{HusonRuppScornavacca10}. 

Arguably the most significant variation is the distinction between the version of the problem on rooted trees (rMAF) and the variant we study here, on unrooted trees (uMAF). Both variants are NP-hard and APX-hard, and this has sparked the development of constant-factor approximation algorithms, fixed-parameter tractable (FPT) algorithms and mathematical programming approaches; we refer
to \cite{WhiddenBZ13,shi2018parameterized,bulteau2019parameterized} for overviews of MAF variations and results.
This research effort has yielded highly efficient exact algorithms for solving rMAF in practice \cite{Whidden2014,yamada2019better,dempsey2024simple}. Such algorithms are typically FPT branching algorithms, strengthened with some basic data reduction rules and bounds.  The research effort around uMAF has taken a slightly different path. FPT branching algorithms are not quite so effective here: due to its somewhat different underlying combinatorial structure, the `natural' branching scheme for uMAF has a higher branching factor than rMAF. Instead, for solving the problem exactly there has been a stronger focus on developing data-reduction rules \cite{kelk2024deep}, and integer linear programming (ILP) \cite{van2022reflections}. 

Recently, Van Wersch et al. \cite{van2022reflections} undertook an experimental study for uMAF of the relative strengths and weaknesses of a FPT branching algorithm and a polynomial-size ILP formulation, both with and without aggressive data reduction. Notably, of the core dataset consisting of 735 pairs of trees, none of the tested methods could solve 38 of the tree pairs within 5 minutes, even after full data reduction. Here we advance the state-of-the-art in exact uMAF solving by applying branch-\&-price to an exponential size ILP formulation. This formulation was originally defined by Olver et al. as scaffolding for the design of a polynomial-time 2-approximation for rMAF \cite{olver2023duality}; they did not, however, attempt to solve the ILP. We observe that the exponential-size ILP maps easily to uMAF and describe a non-trivial dynamic programming algorithm to solve the pricing problem. Specifically, our algorithm solves a weighted version of the well-studied \emph{maximum agreement subtree problem} \cite{steel93} to which our pricing problem forms an instance.

In our experimental section we observe that after data reduction our algorithm can solve 731 of the 735 tree-pairs within 5 minutes (and the remaining 4 tree pairs in 5-7 minutes); the benchmark dataset has thus now been completely solved. The amount of branching is also strikingly low. Interestingly, the
state-of-the-art FPT branching algorithm (still) times out on 19 of the data-reduced 735 instances, and these 19 are all instances that the branch-\&-price algorithm can solve relatively easily even prior to data reduction. Our experimental section also explores a number of further properties of the branch-\&-price algorithm.

We conclude the article with a number of discussion points and open problems.

\section{Preliminaries}
Let $T_1, T_2$ be two unrooted (binary) phylogenetic $X$-trees. We write $T_1 = T_2$ if there is isomorphism between $T_1$ and $T_2$ that preserves the $X$ labels. A subset $X' \subseteq X$ naturally induces an embedding in an unrooted binary $X$-tree $T$:  the unique minimal subtree that connects all the labels in $X'$.  After suppressing degree-2 nodes in this embedding, an unrooted binary phylogenetic $X'$-tree is obtained, denoted $T|X'$. An agreement forest of $T_1, T_2$ is a partition of $X$ into blocks $\{ B_1, \ldots, B_k \}$ such that (1) in $T_1$ the embeddings of the blocks are mutually vertex disjoint, (2) in $T_2$ the embeddings of the blocks are mutually vertex disjoint and (3) for each $B_i$, $T_1|B_i = T_2|B_i$. The size of an agreement forest is simply the number of blocks in it, $k$. An agreement forest of minimum size is called a maximum agreement forest (uMAF) and we often overload the term uMAF to also refer to the computational problem of constructing such an optimal agreement forest. The definition of the problem on rooted trees, rMAF, differs slightly due to the use of rooted trees, but it is not necessary to define this version of the problem for this article.

\section{A branch-\&-price formulation for the unrooted maximum agreement forest problem}
In this section we outline our branch-\&-price approach to the uMAF problem. First, we give an ILP formulation of uMAF that is based on a formulation of rMAF by~\citet{olver2023duality}. Indeed, the authors' ILP formulation of rMAF is valid for uMAF when exchanging rooted for unrooted phylogenetic $X$-trees. Subsequently, we describe the master problem, the pricing problem and branching rules to define our branch-\&-price approach to uMAF. To this end, let $T_1$ and $T_2$ be two unrooted phylogenetic $X$-trees. Let
\begin{align*}
\mathcal{C}= \{ Y \subseteq X: T_1|Y  = T_2|Y \}
\end{align*}
denote the set of all blocks (i.e. subsets) $Y$ that induce the same topology subtree in $T_1$ and $T_2$. Note that any unrooted agreement forest of $T_1$ and $T_2$ can be formed with elements from $\mathcal{C}$. Therefore, we consider binary decision variables
\begin{align*}
a_Y=\begin{cases}1 &\text{if }Y\in\mathcal{C},\\
0 &\text{otherwise.}
\end{cases}
\end{align*}
For $i\in\{1,2\}$ and $Y\subseteq X$, let $V_i[Y]$ denote the internal vertices of the induced subtree of $T_i$ with leafset $Y$. In addition, let $V[Y] = V_1[Y] \cup V_2[Y]$. Then, following the rationale of~\citet{olver2023duality}, the following ILP is valid:

\begin{form}\label{form::uMAF}
\begin{align}
d_{MAF}(T_1, T_2)=~\min~\sum_{Y\in\mathcal{C}}a_Y~~~~~~~~~&\label{obj1}\\
\text{s.t.}~~~~\,\sum_{Y\in\mathcal{C}:\,x\in Y}a_Y&\geq 1~&~&\forall\,x\in X\label{con::cover}\\
\sum_{Y\in\mathcal{C}:\,v\in V[Y]}a_Y&\leq 1~&~&\forall\,v\in V[X]\label{con::packing}\\
a_Y&\in\{0,1\}~&~&\forall\,Y\in\mathcal{C}\nonumber
\end{align}
\end{form}
The constraints~\eqref{con::cover} ensure that each taxon in $X$ is covered by a block~$Y\in\mathcal{C}$. Similarly, we require the blocks $Y$ to be disjoint in both $T_1$ and $T_2$ by constraints~\eqref{con::packing}. Thus, the objective function value~\eqref{obj1} equals the size of an unrooted agreement forest of $T_1$ and $T_2$.

Formulation~\ref{form::uMAF} is well-suited to be the master problem for a branch-\&-price approach because it is constituted of an exponential number of variables and a polynomial number of constraints. The latter suggests that constraints of the following LP dual of Formulation~\ref{form::uMAF} can be separated efficiently:

\begin{form}\label{form::dual}
\begin{align}
\max~\sum_{x\in X}\alpha_x-\sum_{v\in V[X]}\beta_v&\nonumber\\
\text{s.t.}~~\sum_{x\in Y}\alpha_x-\sum_{v\in V[Y]}\beta_v&\leq 1~&~&\forall\,Y\in\mathcal{C}\label{con::dual}\\
\alpha_x&\geq 0~&~&\forall\,x\in X\nonumber\\
\beta_v&\geq 0~&~&\forall\,v\in V[X]\nonumber
\end{align}
\end{form}
In a branch-\&-price scheme we restrict the master problem to subsets of decision variables. This means, when solving the restricted master problem not all constraints~\eqref{con::dual} are enforced, i.e., some might be violated. Hence, we separate constraints~\eqref{con::dual} by identifying the most violated ones and adding the corresponding columns to the primal program/restricted master problem. Establishing dual feasibility in this manner forms the basis for solving the master problem. In other words, given fixed dual variable values, we seek to solve the pricing problem
\begin{align}\label{pricing}
\max_{Y\in\mathcal{C}}~\sum_{x\in Y}\alpha_x-\sum_{v\in V[Y]}\beta_v.
\end{align}
Next we show that problem~\eqref{pricing} can be viewed as an instance of the weighted maximum agreement subtree problem. This will enable us to develop an algorithm to solve~\eqref{pricing} in $\mathcal{O}(n^2)$ time and therefore generate columns for our master problem in polynomial-time.

\subsection{The weighted maximum agreement subtree problem}
Given two unrooted phylogenetic $X$-trees $T_1$ and $T_2$, weights $w(v)$, $v\in V[X]$, and $w(x)$, $x\in X$, the \emph{Weighted Maximum Agreement Subtree Problem (WMAST)} consists of maximizing
\begin{align*}
\text{WMAST}\left(T_1,T_2\right)=\sum_{v\in V[X']}w(v)+\sum_{x\in X'}w(x)
\end{align*}
ranging over all $X' \subseteq X$ such that $T_1|X' = T_2|X'$. Observe that problem~\eqref{pricing} is a WMAST when setting $w(x)=\alpha_x$, $x\in X$, and $w(v)=-\beta_v$, $v\in V[X]$.

For uniform weights $w(x)$, $x\in X$, and $w(v)=0$ for all $v\in V[X]$, the WMAST reduces to an unweighted problem which was introduced by~\citet{finden1} and solved by~\citet{steel93} in $\mathcal{O}(n^2)$ time with a dynamic programming algorithm, respectively. Notice that the latter result also holds for the version of the WMAST in which we require trees $T_1$ and $T_2$ to be rooted with leaf-sets $X_1$ and $X_2$, respectively. We denote this rooted version by rWMAST\footnote{The definition is the same as for WMAST, except that for $T_1|X' = T_2|X'$ to hold the direction of the arcs in $T_1, T_2$ also needs to be respected.}. In this section we extend the method of~\citet{steel93} to arbitrary but fixed weights.

First, for $i\in\{1,2\}$, observe that the removal of one edge~$e$ from the unrooted phylogenetic $X$-tree $T_i=(V_i,E_i)$ yields a rooted phylogenetic $X_i$-tree $t_i^e$ with $X=X_1\cup X_2$. Hence,
\begin{align}
&\text{WMAST}\left(T_1,T_2\right)\nonumber\\
&=\max_{e\in E_1,~f\in E_2}\left\{~\text{rWMAST}\left(t_1^e,t_1^f\right)+\text{rWMAST}\left(t_2^e,t_2^f\right),\right.\nonumber\\
&~~~~~~~~~~~~~~~~~~~~\,\left.\text{rWMAST}\left(t_1^e,t_2^f\right)+\text{rWMAST}\left(t_2^e,t_1^f\right)\right\}.\label{wmast_relation}
\end{align}
Therefore, we focus on solving rWMAST. To this end, let $t_1$ and $t_2$ be rooted phylogenetic trees with roots $\rho_1$ and $\rho_2$, internal vertices $V[t_1]$ and $V[t_2]$, and leaf sets $X[t_1]$ and $X[t_2]$, respectively. Let $w(v)$, $v\in V[t_i]\cup X[t_i]$, $i\in\{1,2\}$ be weights, and let $t_{ij}$, $i,j\in\{1,2\}$, denote the rooted phylogenetic trees we obtain by removing $\rho_i$ from $t_i$. Then, for $i\in\{1,2\}$, $x\in X[t_i]$, define $P_i(x)$ as the path from $\rho_i$ to $x$ in $t_i$ and consider the quantities $\mathcal{V}\left(t_1,t_2\right)$, $\mathcal{M}\left(t_1,t_2\right)$, $\mathcal{W}\left(t_1,t_2\right)$ defined as follows: if $|X[t_i]|\geq 2$, $i\in\{1,2\}$, then
\begin{align*}
&\mathcal{V}\left(t_1,t_2\right)=\max\limits_{j\in\{1,2\}}\left\{w(\rho_1)+w(\rho_2)+\mathcal{W}\left(t_1,t_2\right),\right.\\
&~~~~~~~~~~~~~~~~~~~~~~~~~~~\,\left.w(\rho_1)+\mathcal{V}\left(t_{1j},t_2\right),\,w(\rho_2)+\mathcal{V}\left(t_1,t_{2j}\right)\right\},\\
&\mathcal{M}\left(t_1,t_2\right)\\
&=\max\limits_{j\in\{1,2\}}\left\{w(\rho_1)+w(\rho_2)+\mathcal{W}\left(t_1,t_2\right),\,\mathcal{M}\left(t_1,t_{2j}\right),\,\mathcal{M}\left(t_{1j},t_2\right)\right\},\\
&\mathcal{W}\left(t_1,t_2\right)\\
&=\max\left\{\mathcal{V}\left(t_{11},t_{21}\right)+\mathcal{V}\left(t_{12},t_{22}\right),\,\mathcal{V}\left(t_{11},t_{22}\right)+\mathcal{V}\left(t_{12},t_{21}\right)\right\}.
\end{align*}
If $t_1$ or $t_2$ is a singleton with $X[t_1]\cap X[t_2]=\{x\}$, then $$\mathcal{V}\left(t_1,t_2\right)=\,\text{rWMAST}\left(P_1(x),P_2(x)\right)~~\text{and}~~\mathcal{M}\left(t_1,t_2\right)=w(x).$$ Otherwise, $\mathcal{V}\left(t_1,t_2\right)=\mathcal{M}\left(t_1,t_2\right)=0$. Observe that both functions $\mathcal{V}$ and $\mathcal{M}$ evaluate the weights of vertices which form an agreement forest for inputs $t_1$ and $t_2$. Then, we arrive at the following proposition: 

\begin{prop}\label{prop::rWMAST}
For $i\in\{1,2\}$, let $t_i$ be a rooted phylogenetic $X_i$-tree. Let $\rho$ be the root of the agreement subtree corresponding to $\mathcal{M}\left(t_1,t_2\right)$. Then, 
\begin{align*}
\text{rWMAST}\left(t_1,t_2\right)=\mathcal{M}\left(t_1,t_2\right)+\max\left\{0,\sum_{v\in P_1(\rho)\cup P_2(\rho)}w(v)\right\}.
\end{align*}
\end{prop}
\begin{proof}
By induction on the number of leaves $l$ in $t_1$ and $t_2$. First, let $l=2$. Since both $t_1$ and $t_2$ are not empty, they are both singletons. Denote the singleton vertices in $t_1$ and $t_2$ by $x$ and $y$, respectively. Then, $\mathcal{M}(t_1,t_2)=w(x)$ if $x=y$ and $\mathcal{M}(t_1,t_2)=0$ otherwise. Hence, our claim holds. Assume our induction hypothesis holds for $l$ leaves and $t_1,t_2$ have $l+1$ leaves between them.
\begin{description}
\item[Case 1:] $\mathcal{M}\left(t_1,t_2\right)=w(\rho_1)+w(\rho_2)+\mathcal{W}\left(t_1,t_2\right)$. Let $t$ denote the rooted agreement subtree of $t_1$ and $t_2$ corresponding to $\mathcal{M}\left(t_1,t_2\right)$. Without loss of generality we have $\mathcal{W}\left(t_1,t_2\right)=\mathcal{V}\left(t_{11},t_{21}\right)+\mathcal{V}\left(t_{12},t_{22}\right)$. Otherwise exchange $t_{21}$ and $t_{22}$. Observe that the rooted agreement subtree $t$ arises from the concatenation of rooted agreement subtrees $s_1$ and $s_2$ of $t_{11}$ and $t_{21}$, and $t_{12}$ and $t_{22}$, respectively. For $j\in\{1,2\}$, by the definition of function $\mathcal{V}$, $s_j$ contains the root of $t_{1j}$ and $t_{2j}$, and is constructed from rooted agreement subtrees in $s_j$. Eventually this recursive process identifies subtrees $r_{1j},\dots,r_{pj}$, $j\in\{1,2\}$, for some integer $p\geq 1$ such that
\begin{align*}
\mathcal{V}\left(r_{i1},r_{i2}\right)&=\mathcal{M}\left(r_{i1},r_{i2}\right)~&~&\forall\,i\in\{1,\dots,p\}.
\end{align*}
The corresponding rooted agreement subtrees are maximum with respect to the given weights by our induction hypothesis, i.e.,
\begin{align*}
\text{rWMAST}\left(r_{i1},r_{i2}\right)&=\mathcal{V}\left(r_{i1},r_{i2}\right)~&~&\forall\,i\in\{1,\dots,p\}.
\end{align*}
This means, by the maximality of $\mathcal{V}$, subtree $t$ is maximum among rooted weighted agreement subtrees which contain roots $\rho_1$ and $\rho_2$. Since a maximum rooted agreement subtree of $t_1$ and $t_2$ has to contain roots $\rho_1$ and $\rho_2$ by the maximality of $\mathcal{M}$, we conclude that $t$ is a rooted weighted maximum agreement subtree of $t_1$ and~$t_2$.
\item[Case 2:] For $j\in\{1,2\}$, $\mathcal{M}\left(t_1,t_2\right)=\mathcal{M}\left(t_1,t_{2j}\right)$ or $\mathcal{M}\left(t_1,t_2\right)=\mathcal{M}\left(t_{1j},t_2\right)$. W.l.o.g. $\mathcal{M}\left(t_1,t_2\right)=\mathcal{M}\left(t_1,t_{21}\right)$. By our induction hypothesis, $\text{rWMAST}\left(t_1,t_{21}\right)=\mathcal{M}\left(t_1,t_{21}\right)$. Thus, our claim follows from the maximality of $\mathcal{M}(t_1,t_2)$.
\end{description}
\end{proof}

Observe that the maximization~\eqref{wmast_relation} is over $\mathcal{O}(n^2)$ terms. This means, we can define a partial order over these terms by inclusion and store its linear extension in $\mathcal{O}(n^2)$ space. Then, along this linear order we can evaluate the recursive functions $\mathcal{V}$, $\mathcal{M}$ and $\mathcal{W}$ iteratively. Thus, we conclude from Proposition~\ref{prop::rWMAST} that we can solve the WMAST in $\mathcal{O}(n^2)$ time. As we have observed before this means we can solve the pricing problem~\eqref{pricing} in $\mathcal{O}(n^2)$ time, too.

\begin{figure}[!t]
\centering
\includegraphics[scale=0.4]{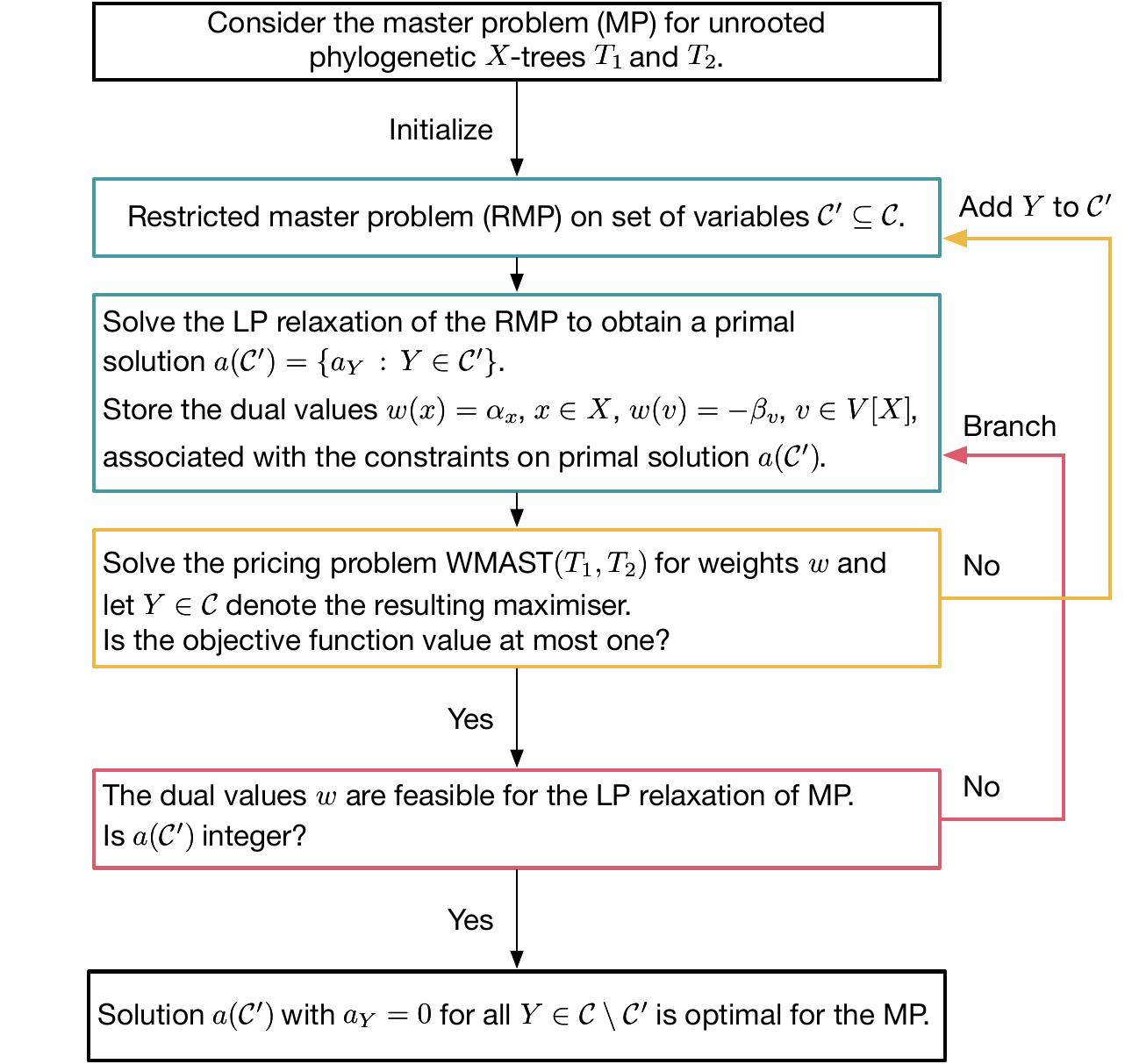}
\caption{A branch-\&-price scheme for uMAF. Intermediate solutions, the branching and pricing are shown in blue, red and yellow, respectively.}\label{fig::scheme}
\end{figure}

\subsection{Implementation of the branch-\&-price algorithm}
\label{subsec:implement}
We summarize our discussion of a branch-\&-price algorithm for the uMAF problem in Figure~\ref{fig::scheme}. While Formulation~\ref{form::uMAF} serves as the master problem, we initalize the restricted master problem by choosing $\mathcal{C}'$ to be the set of singleton leafsets in~$\mathcal{C}$. This allows us to solve the RMP in polynomial time and ensures that a feasible primal solution exists. Subsequently, we solve the pricing problem in quadratic time to generate columns for the RMP. Eventually, we reach dual feasibility (in the worst case for $\mathcal{C}'=\mathcal{C}$). If in addition our primal solution is integer, then we conclude that our optimal solution for the RMP is an optimal solution for the MP, too. Otherwise, we need to branch on variables $a(\mathcal{C}')$, i.e., enforce some of the integer constraints of the MP. We will consider two different implementation settings of our branch-\&-price algorithm by specifying two branching strategies: set $a_Y=1$ such that a selection criterion $S(Y)$ is maximum among all $Y\in\mathcal{C}'$. For our first strategy we choose $S(Y)=|Y|$ and for our second strategy we define $S(Y)=|Y|/|V[Y]|$\footnote{As discussed later in the experimental section, there was ultimately insufficient branching in the reference dataset to be able to say that one branching strategy is superior to the other, so the experiments simply used the second strategy.}.

Preliminary experiments have shown that a slight perturbation of solutions to the pricing problem improves (on average) the performance of our branch-\&-price implementations. Specifically, for some small $\epsilon>0$, we consider the following dual program instead of Formulation~\ref{form::dual}:
\begin{form}\label{form::dual2}
\begin{align}
\max~\sum_{x\in X}\alpha_x-\sum_{v\in V[X]}\beta_v&-\epsilon\cdot m_Y\nonumber\\
\text{s.t.}~~\sum_{x\in Y}\alpha_x-\sum_{v\in V[Y]}\beta_v&\leq 1~&~&\forall\,Y\in\mathcal{C}\nonumber\\
\alpha_x&= 1~&~&\forall\,x\in X\nonumber\\
m_Y&\geq \beta_v~&~&\forall\,v\in V[Y],\,Y\in\mathcal{C}\nonumber\\
\beta_v&\geq 0~&~&\forall\,v\in V[X]\nonumber
\end{align}
\end{form}
Then, the LP dual of Formulation~\ref{form::dual2} is as follows:
\begin{form}\label{form::uMAF2}
\begin{align}
\min~\sum_{Y\in\mathcal{C}}a_Y+\sum_{x\in X}b_x~~~~~~~&\label{obj2}\\
\text{s.t.}~~~~~~~~\,\sum_{Y\in\mathcal{C}:\,x\in Y}a_Y+b_x&\geq 1~&~&\forall\,x\in X\label{con::a-b}\\
\sum_{Y\in\mathcal{C}:\,v\in V[Y]}\left(a_Y+c_{v,Y}\right)&\leq 1~&~&\forall\,v\in V[X]\label{con::c1}\\
\sum_{v\in V[Y]}c_{v,Y}&\leq \epsilon~&~&\forall\,Y\in\mathcal{C}\label{con::c2}\\
a_Y&\geq 0~&~&\forall\,Y\in\mathcal{C}\nonumber\\
b_x&\in\mathbb{R}~&~&\forall\,x\in X\nonumber\\
c_{v,Y}&\geq 0~&~&\forall\,v\in V[Y],\,Y\in\mathcal{C}\nonumber
\end{align}
\end{form}
Observe that Formulation~\ref{form::uMAF2} reflects an extension of the LP relaxation of Formulation~\ref{form::uMAF} by variables $b_x$ and $c_{v,Y}$ but the change in constraints does not affect the optimal solution. To see this, first notice that constraints~\eqref{con::a-b} combined with objective function~\eqref{obj2} yield $b_x\in[0,1]$ because $a_Y\in[0,1]$. Hence, $b_x$ encodes $\{x\}\in\mathcal{C}$ in the same way $a_{\{x\}}$ does, leaving the optimum objective function value unchanged. Secondly, constraints~\eqref{con::c2} force variables $c_{v,Y}$ to take very small values, making it easy to recover integer values for $a_Y$ from the perturbations by $c_{v,Y}$ in constraints~\eqref{con::c1}. Thus, replacing pricing problem~\eqref{pricing} with
\begin{align}\label{pricing2}
\max_{Y\in\mathcal{C}}~|X|-\sum_{v\in V[Y]}\beta_v-\epsilon\cdot m_Y
\end{align}
lets our branch-\&-price algorithm terminate with an optimal solution to the uMAF problem. Throughout our computational experiments we keep pricing problem~\eqref{pricing2} instead of~\eqref{pricing}.

\section{Computational experiments}

The branch-\&-price algorithm itself was implemented in Java
using the jORlib package version 1.1.1 interface using CPLEX 22.1.1
as its ILP solver. For auxiliary data structure manipulation the JGraphT package version
0.9.0 was used. All experiments were conducted on a MacBook Air with the M3 chip, featuring an 8-core CPU and a 10-core GPU, 24GB of unified memory, and
running macOS Sonoma (version 14.4) using Darwin Kernel (version 23.4.0).
We have publicly released the implementation. The link to the GitHub page containing the branch-\&-price algorithm can be found at \url{https://github.com/simivych/MAF_B-P}.

\subsection{The reference dataset of Van Wersch et al.}

To test our implementation we re-analysed a reference dataset designed by \cite{van2022reflections}. This
consists of 735 pairs of unrooted binary phylogenetic $X$-trees. Briefly, it consists of 5  tree pairs for every $(t,s,k)$ combination, where
\begin{itemize}
\item $t \in \{50,100,150,200,250,300,350\}$,
\item $s \in \{50, 70,90\}$,
\item $k \in \{5,10,15,20,25,30,35\}$.
\end{itemize}
Each $(t,s,k)$ combination yields a tree pair as follows: a tree $T_1$ is randomly sampled on $t$ taxa, with ``skew'' $s$ (where $s=50$ means an almost balanced tree, and $s=90$ is an almost path-like tree), and then $k$ random tree bisection and reconnect (TBR) moves are applied to it to obtain a second tree $T_2$. A TBR move is a specific type of topological rearrangement of the tree: a central result in the literature states that the minimum number of TBR moves required to turn one tree into another, is equal to the size of an uMAF, minus 1 \cite{steel01}.
The application of these random TBR moves therefore 
ensures that an uMAF of $T_1$ and $T_2$ has \emph{at most} $k+1$ blocks. (It can be  strictly less than $k+1$ if some of the randomly applied TBR moves cancel each other out). According to experiments \cite{van2022reflections} the skew parameter $s$ has only a very limited influence on overall algorithmic performance, so we do not analyse this experimental parameter further here.

In \cite{van2022reflections}, these 735 tree pairs were used to compare the performance of two solution techniques, an FPT branching algorithm (by Whidden et al. \cite{whidden2018calculating}, leveraging an optimization from \cite{ChenFS15}) and a new polynomial-size ILP formulation; in both cases both with and without polynomial-time kernelization. Kernelization is an established pre-processing technique in parameterized complexity which reduces the size of the instances - in this case, the number of leaves - without changing the size of the optimum, or changing it in a predictable way \cite{kernelization2019}. The largest tree pairs after kernelization had 174 leaves, and the time to perform kernelization is negligible compared to solving time for uMAF. Due to the performance guarantee of the  pre-processing used, the kernelized instances have a number of leaves at most 11 times the true size of the uMAF; for many of the instances this factor is somewhat smaller. Roughly speaking, therefore, a $(t,s,k)$ tree pair
has $\Theta(k)$ leaves after kernelization. The FPT branching algorithm has a running time of $O( 3^{OPT}\cdot \text{poly}(|X|) )$, where OPT here refers to the size of an uMAF, and the polynomial-size ILP formulation has $O(|X|)$ binary variables and $O(|X|^4)$ constraints.

\subsection{Time to generate columns, and number of columns generated}
Table \ref{tab:taxa} is derived from the performance of the branch-\&-price algorithm on the non-kernelized instances. We focussed on non-kernelized instances here because the number of leaves in these instances can be experimentally controlled, in contrast to the kernelized instances where this is linearly dependent on the size of the uMAF. The table shows that the number of leaves in an instance has a large impact on the performance of the branch-\&-price algorithm, primarily because of the extra time required to generate each column (i.e. to solve the pricing problem; recall that this has $O(t^2)$ complexity);  it increases from 0.36s for $t=50$ to 16.57s for $t=200$. However, the increasing number of leaves also contributes to an increase in the number of additional columns required overall, both in absolute terms (32 for $t=50$, but $55$ for $t=200$) and also in relative terms (1.90 times the true size of the uMAF for $t=50$, compared to 2.84 for $t=2.84$).

\begin{table*}
\centering
    ~~~~~\begin{tabular}{|p{3cm}|p{3cm}|p{3cm}|p{3cm}|p{3cm}|p{0.1cm}}
     \cline{1-5}
        \centering Number of leaves in tree $t$ & \centering Average time per column (s) & \centering Average number of additional columns generated & \centering Average additional columns generated / size of uMAF & Number of instances that branched  \\ \cline{1-5}
        \centering 50 & \centering 0.36 & \centering 31.81 & \centering 1.90 &  \centering 3/105 &  \\  \cline{1-5}
        \centering 100 & \centering 2.34 & \centering 38.11 & \centering 1.97 & \centering 0/105 & \\  \cline{1-5}
        \centering 150 & \centering 6.85 & \centering 42.20 & \centering 2.13 & \centering 0/105 & \\  \cline{1-5}
        \centering 200 & \centering 16.57 & \centering 55.13 & \centering 2.84 & \centering 0/105 & \\  \cline{1-5}

    \end{tabular}
    \caption{Performance of the branch-\&-price algorithm on non-kernelized instances of the reference dataset with up to 200 leaves. The number of columns generated, does not include the $t$ trivial columns (corresponding to single leaves) included in the initial solution.}
    \label{tab:taxa}
\end{table*}
The compound effect of these phenomena made it difficult for the branch-\&-price algorithm to solve non-kernelized instances with more than 200  taxa within 5 minutes, which is why the table only considers $t \leq 200$. Indeed, up to 100 leaves the branch-\&-price algorithm solves all instances
well within 5 minutes, and solving time is on average around 100 seconds. However, between 100 and 200 taxa the solving time increases roughly linearly, and the failure rate also increases linearly, to the point that at 150 taxa approximately 50$\%$ of the instances failed, at 200 taxa around 90$\%$, and by 250 taxa the failure rate is 100$\%$. Overall, the branch-\&-price algorithm was capable of solving (only) 280 of the 735 instances in 5 minutes. 

These results indicate that when using the branch-\&-price algorithm it is potentially useful to apply polynomial-time kernelization (or other forms of fast pre-processing) to first reduce the number of leaves in the instance.  We confirm this hypothesis in the next section. In any case, we observed that even on non-kernelized instances, the branch-\&-price algorithm exhibited some `islands of superiority' compared to the algorithms tested in \cite{van2022reflections}. Specifically: branch-\&-price could solve \emph{all} non-kernelized instances with up to 100 taxa comparatively easily. In contrast, as the termination experiments in the next section indicate, the FPT branching algorithm has multiple failure instances within this region (even on kernelized instances).

We briefly return to the number of additional columns generated (beyond the initial $t$ singleton columns), and the ratio of this number to uMAF. As shown in Table \ref{tab:taxa} this ratio rises from 1.9 to 2.84. A secondary analysis (not shown here) shows that if the same 420 non-kernelized trees with $t \leq 200$ are instead organized by $k$, the ratio seems to approach 2 from above as $k$ increases, but this phenomenon requires further research.

\subsection{Termination} An interesting finding of \cite{van2022reflections} was that there were 38 instances which neither of the two tested algorithms could solve to optimality in 5 minutes, even on the reduced, kernelized instances.

We applied our branch-\&-price algorithm on the kernelized versions of the 735 tree pairs and found that within 5 minutes, 731 could be solved to optimality. The 4 instances that could not be solved to optimality in 5 minutes, were solved in 5-7 minutes. As a result of this, the reference dataset has now been completely solved. We have made this information (i.e. the size of the uMAF for every tree pair) available on our aforementioned GitHub page. 

To mitigate the fact that the laptop used in \cite{van2022reflections} is somewhat older, we also re-ran the FPT branching algorithm on our laptop on the kernelized instances. The FPT branching algorithm could not solve 19 instances within 5 minutes. Of these 19, 8 had still not terminated after 8 minutes, and the other 11 terminated in 5-7 minutes. (We chose not to re-run the polynomial-size ILP, as this is not capable of solving tree pairs with more than 70 leaves after kernelization, due to the ILP itself becoming enormous: recall that it has $O(|X|^4)$ constraints.) Although the sample of non-terminating instances is small, it is instructive to note that the 19 instances on which the FPT branching algorithm failed, all had $k \in \{30,35\}$ and prior to kernelization at most 50 leaves. This is consistent with the fact that by design the algorithm runs in time exponential in the size of the uMAF, and has only (low-order) polynomial dependency on the number of leaves. The 4 kernelized instances that the branch-\&-price algorithm failed on all had a large number of leaves, between 137 and 174, putting them in the top 5$\%$ in terms of number of leaves. (Prior to kernelization these trees all had $t \in \{300,350\}$ and $k \in \{30,35\}$). This corroborates the findings in the previous section i.e. that the performance of the branch-\&-price algorithm is heavily sensitive to the number of leaves in the input trees.

\subsection{Branching}
As can be seen in Table \ref{tab:taxa}, only 3 of the 420 non-kernelized instances with $t \leq 200$ branched. (One of these used 3 branching nodes in total, one used 5, and one used more than 5.) On the kernelized dataset, only 4 branched out of the entire 735 instances. This suggests that, at least on this dataset, the root LP-relaxation is extremely often integral, and the time spent on any branching is completely negligible compared to the time required to generate the columns. The solving power of CPLEX was thus ultimately unnecessary: any efficient LP solver with basic support for branching would also have worked well. Due to the lack of branching, we were not able to meaningfully explore the strengths and weaknesses of the two different branching strategies (see Section \ref{subsec:implement}). For this reason we only used the second strategy throughout our experiments, i.e.  the one that selects the column with the largest size-of-leafset-to-number-of-internal-nodes ratio.

\subsection{Solving time}
We do not provide detailed results but observe that, on the kernelized dataset, the average solve time for the branch-\&-price algorithm rises slowly from negligible (for $k \leq 10$) to around 120s (for $k=35$); recall that such instances have $\Theta(k)$ leaves. Regarding comparison to other methods: as noted earlier the solving time of the FPT branching algorithm is exponentially dependent on the size of the uMAF. When uMAF is not too large (no more than 15-18, say) the FPT branching algorithm tends to solve more quickly than the branch-\&-price algorithm. This is not unexpected: the FPT branching algorithm has a very lightweight design, compared to the overheads involved in the branch-\&-price algorithm (such as initiating an LP solver) and the time to generate columns. Nevertheless, as noted above, the branch-\&-price algorithm runs more quickly (and terminates more often) once high uMAF starts to cause exponential slowdown in the FPT branching algorithm.  

\section{Discussion and conclusion}
When kernelization is applied the branch-\&-price algorithm is capable of solving (almost) all instances in the reference data-set in 5 minutes. Interestingly, even without kernelization the algorithm outperforms a state-of-the-art FPT branching algorithm on several (for the FPT algorithm) challenging instances with at most 100 leaves. With kernelization the branch-\&-price algorithm terminates more often than the FPT algorithm, although the FPT algorithm tends to run more quickly than the branch-\&-price algorithm on easier instances. Clearly, the time to solve the pricing algorithm, which grows sharply as the number of leaves increases, is a bottleneck. We expect that speeding the pricing algorithm up will greatly improve performance given that the time spent on branching is negligible. Indeed, why is the amount of branching so low? Also, is there any mathematical significance to the insight that, very often, the number of additional columns generated is roughly twice the size of the uMAF; could there be a link with existing (polyhedral-based) approximation algorithms such as \cite{olver2023duality}?

Looking forward, there is now a need to generate a more challenging reference dataset. This will allow a deeper exploration of the phenomena observed in this article and constitute a basis for developing a more sophisticated branch-\&-price algorithm. It is particularly interesting to determine for which types of instances branching becomes a non-negligible part of the running time.

Finally, we note that our branch-\&-price approach could also be implemented for rMAF, since the ILP has the same overall structure. Interestingly,  our dynamic-programming algorithm for the pricing problem becomes simpler in the case of rooted trees because the maximation~\eqref{wmast_relation} can be discarded. Furthermore, the existence of fast solution algorithms for the unweighted rWMAST~\cite{cole00} suggests further improvements of the pricing procedure which are not applicable for unrooted trees. Given the bottleneck formed by solving the pricing problem, this could constitute an important speedup. It is natural to then explore whether a  branch-\&-price algorithm for rMAF can be competitive against the heavily optimized FPT branching algorithms that exist for this problem. Extensions to multiple-tree variants of uMAF and rMAF can also be investigated.

\section{Acknowledgements}
M. Frohn was supported by grant OCENW.M.21.306 from the Dutch
Research Council (NWO). 






\bibliographystyle{elsarticle-num-names}







\end{document}